\documentclass[twoside]{article}
\usepackage{amsmath,amssymb,amsfonts,amsthm}
\usepackage{color}
\usepackage{graphicx} 
\usepackage{authblk}

\usepackage[showzone=true,showisoZ=false]{datetime2}
\DTMsetdatestyle{mmddyyyy}
\DTMsetup{datesep={/}}

\date{\today}

\newif\ifdraft

\ifdraft
\usepackage[notcite,notref]{showkeys}
\date{DRAFT: \today}
\fi

\usepackage{hyperref}
\hypersetup{
	colorlinks=true,
	linkcolor=blue,
	citecolor=green
}
\newcommand{\myrunningheads}
{\ifdraft
\pagestyle{myheadings}
\markboth
{DRAFT \today\quad \DTMcurrenttime\ EST} 
{DRAFT \today\quad \DTMcurrenttime\ EST} 
\else
\pagestyle{myheadings}
\markboth
{Waves in particle lattices with long-range forces}
{B. Ingimarson and R. L. Pego}
\fi
}
\myrunningheads

\newcommand{\nwc}{\newcommand}
\nwc{\fix}[1]{\textcolor{red}{[#1]}}
\nwc{\note}[1]{\textcolor{blue}{#1}}

\nwc{\oldnu}{\zeta}  

\nwc{\pow}{\alpha}
\nwc{\bshift}{b}  
\nwc{\prate}{p}
\nwc{\vp}{\varphi}
\nwc{\intR}{\int_{-\infty}^\infty}
\nwc{\fpl}{\frac{\pi}{L}}
\nwc{\ict}{{\sigma}}
\nwc{\mict}{{\bar\sigma}}

\nwc{\R}{{\mathbb{R}}}
\nwc{\N}{{\mathbb{N}}}
\nwc{\Z}{{\mathbb{Z}}}
\nwc{\C}{{\mathbb{C}}}
\nwc{\D}{\partial}
\nwc{\eps}{\epsilon}
\nwc{\calL}{{\mathcal L}}
\nwc{\calF}{{\mathcal F}}
\DeclareMathOperator{\sgn}{sgn}
\DeclareMathOperator{\sinc}{sinc}

\DeclareMathOperator{\sech}{sech}
\DeclareMathOperator{\im}{Im}
\DeclareMathOperator{\re}{Re}
\DeclareMathOperator{\diag}{diag}

\nwc{\coloneq}{\colon=}
\nwc{\inv}{^{-1}}
\nwc{\dds}[1]{\frac{d#1}{ds}}
\nwc{\qua}{{\ \ }}
\newtheorem{theorem}{Theorem}

\newtheorem{corollary}[theorem]{Corollary}

\newtheorem{lemma}[theorem]{Lemma}
\newtheorem{proposition}[theorem]{Proposition}

\theoremstyle{remark}
\newtheorem{remark}{Remark}

\begin{document}

\title{On long waves and solitons in particle lattices with forces of infinite range}
\author{Benjamin Ingimarson\footnote{Email address: bwi@andrew.cmu.edu} }
\author{Robert L. Pego\footnote{Email address: rpego@cmu.edu} }

\affil{%
Department of Mathematical Sciences\\ 
Carnegie Mellon University\\
Pittsburgh, PA 15213.}

\date{December 15, 2023 (revised)}

\maketitle
\begin{abstract}
We study waves on infinite one-dimensional lattices of particles 
that each interact with all others through power-law forces $F \sim r^{-\beta}$. 
The inverse-cube case corresponds to Calogero-Moser systems which 
are well known to be completely integrable for any finite number of particles.
The formal long-wave limit for unidirectional waves in these
lattices is the Korteweg-de Vries equation if $\beta>4$, 
but with $2<\beta<4$ it is a nonlocal dispersive PDE that 
reduces to the Benjamin-Ono equation for $\beta=3$.
For the infinite Calogero-Moser lattice, we find explicit formulas
that describe solitary and periodic traveling waves.
\end{abstract}


\section{Introduction}

In this work we study wave motions in infinite lattices of 
particles that each interact with all the others through long-range 
power-law forces.  The particle positions $x_j$ are required to 
increase with $j$ and  evolve according to the equations 
\begin{equation}\label{e:sys1}
    \ddot x_j = 
    -\pow \sum_{m=1}^\infty  
    \Bigl(  (x_{j+m}-x_j)^{-\pow-1} -(x_j-x_{j-m})^{-\pow-1}  \Bigr)\,,
\end{equation}
where $\pow>1$. For $\pow=2$ this is an infinite-lattice
version of the  famous Calogero-Moser system \cite{Calogero.71,Moser.75} 
\begin{equation}\label{e:cm1}
    \ddot x_j = \sum_{k\ne j} \frac{2}{(x_j-x_k)^{3}}\,,
\end{equation}
which is well-known to be completely integrable 
and has been extensively investigated when the number of particles is finite.

Wave motions have been widely examined in infinite particle
lattices  with nonlinear {\em nearest-neighbor} forces, 
known as Fermi-Pasta-Ulam-Tsingou (FPUT) lattices.
Such lattices  typically admit a Korteweg-de Vries  scaling limit 
for the unidirectional propagation of long waves of small amplitude,
a fact that helped to trigger the great bounty of discoveries
in the theory of completely integrable systems 
that has emerged over the last half-century \cite{Zabusky.Kruskal.65}.

Also, FPUT lattices typically admit exact solitary wave solutions \cite{Toda-book,Friesecke.Wattis.94,Friesecke.Pego.99}. The form of these waves is known 
explicitly only in the case of the Toda lattice, which is completely integrable.
Recently Vainchtein \cite{Vainchtein.22} surveyed work on solitary waves 
in lattices, including lattices with next-nearest-neighbor or longer-range interactions.  
In particular, existence theorems for interactions of any finite range
were proved recently by 
Herrmann and Mikikits-Leitner \cite{Herrmann.Mikikits-Leitner.16} 
using a KdV approximation argument, and by 
Pankov \cite{Pankov.19} using variational methods.
The former authors mention that the approximation
argument should work for infinite-range interactions if their strength decays rapidly enough,
e.g., exponentially fast. 

Strong motivation for considering lattice systems with power-law forces such as \eqref{e:sys1}
comes from experimental work on solitary waves in chains of repelling magnets
by Moler\'on~{\em et al.}~\cite{MoleronEA2014}. 
These authors mention that long-range dipole-dipole interactions between magnets
separated by a large distance $d$ involve repulsive forces proportional 
to  $d^{-4}$ in theory.  Over distances appropriate to their experiments, however, measurements 
better fit a force law proportional to $d^{-\beta}$ with $\beta\approx 2.73$.
The Calogero-Moser force law, with $\beta=3$, may be considered a reasonable approximation. 
And since such power-law forces have long range, it is interesting to consider the 
infinite-range limit represented by \eqref{e:sys1}.
Admittedly, the system~\eqref{e:sys1} is not a perfect model for the experiment setup of \cite{MoleronEA2014},
not only because dissipation is neglected, but
because a given magnet successively repels and attracts others along the chain due to the alternating orientation of north and south poles. 
Such forces can be treated as differences between 
forces from two systems of repulsive forces, though, and we will discuss this.  
Studying the system~\eqref{e:sys1} is clearly an important step anyway toward understanding 
more general systems with forces of infinite range.

\subsection*{Formal long-wave scaling limits}
As it turns out, a formal KdV limit is possible for the system \eqref{e:sys1} 
with power-law forces of infinite range, but only  when $\pow$ is sufficiently large, 
namely when $\beta=\pow+1\ge4$ as we show below.
When $2<\beta<4$, we find instead in Section~\ref{s:longwave} that a different scaling limit obtains,
with small long waves formally governed by a {\em nonlocal} dispersive PDE of the 
form 
\begin{equation} \label{e:BObeta}
    \D_t u +  u \D_x u +  H|D|^\pow u = 0 \,.
\end{equation}
Here
$H$ is the Hilbert transform, and $|D|^\pow$ has Fourier symbol $|k|^\pow$,
thus the dispersion term $f=H|D|^\pow u$ has Fourier transform 
$\hat f(k)=(-i\sgn k)|k|^\pow \hat u(k)$.
For the case $\pow=2$ corresponding to the infinite Calogero-Moser lattice in particular,
\eqref{e:BObeta} is the {\em Benjamin-Ono} equation, in the form  
\begin{equation}
\D_t u + u \D_x u - H\D_x^2 u = 0 \,.
\end{equation}

There is a well-known link between the Calogero-Moser system and 
Benjamin-Ono equations through the pole dynamics of rational solutions 
\cite{Airault.McKean.ea.77,Choodnovsky.Choodnovsky.77,Case.78,Stone.Anduaga.ea.08}.
Also through pole dynamics, formal continuum limits of Calogero-Moser systems have been 
connected with coupled Benjamin-Ono-type equations 
in the physics literature \cite{Polychronakos.95,Stone.Anduaga.ea.08,Abanov.Bettelheim.ea.09}.
To our knowledge, however, the long-wave limit that we consider herein has 
not been previously described.

\subsection*{Formulae for Calogero-Moser waves}
The fact that dispersive PDE of the form in \eqref{e:BObeta} admit solitary wave 
solutions is a consequence of the analyses of  
Benjamin {\em et al.}~\cite{Benjamin.Bona.ea.90} and Weinstein \cite{Weinstein.87}.
For the long-range particle system \eqref{e:sys1}, 
a rigorous analysis of existence for solitary waves
is out of the scope of the present paper. 
It is plausible, though, that such an analysis 
could be performed by methods like those used for FPUT lattices 
and lattices with longer-range interactions, either of variational 
character \cite{Friesecke.Wattis.94,Pankov.19,Pego.Van.19} 
or of iterative/fixed-point character 
\cite{Friesecke.Pego.99,Herrmann.10,Herrmann.Mikikits-Leitner.16}.

At present, we focus  discussion of solitary and periodic traveling waves 
to the special case of  the infinite Calogero-Moser lattice. 
Waves traveling to the right in such a lattice are solutions with the property
that after some time delay $\tau>0$, the configuration of the lattice
recurs with an index shift and a spatial shift $h>0$, so that
\begin{equation}\label{e:wavesymmetry}
    x_{j+1}(t+\tau)= x_j(t) + h 
\end{equation}
for all $j$ and $t$. 
This means that traveling waves can be expressed in the form
\begin{equation}\label{e:xvp0}
    x_j(t) = jh - \vp(jh-ct), 
\end{equation}
where $c=h/\tau$ and $-\vp(-ct)=x_0(t)$ for all $t$.
Moreover, by the scaling $x_j\mapsto h x_j$, $t\mapsto h^2t$ which leaves \eqref{e:cm1} 
invariant, and a choice of origin for space and time, we can suppose $h=1$ and $x_0(0)=0$.

By making use of B\"acklund transforms for 
Calogero-Moser-Sutherland systems (see \cite{Wojciechowski.82,Wojciechowski.83} 
and also \cite{Abanov.Bettelheim.ea.09,Stone.Anduaga.ea.08,Philip.2019}), 
we have managed to derive striking explicit formulas that determine
both solitary waves and periodic waves for Calogero-Moser lattices.
\begin{theorem}[Solitary waves]\label{t:soliton}
    For each wave speed $c$ satisfying $c^2>\pi^2$, the infinite Calogero-Moser lattice
    admits a solitary wave solution of the form 
\begin{equation}\label{e:xvp}
    x_j(t) = j - \vp(j-ct), 
\end{equation}
where $\vp=\vp(s)$ increases from $\vp(-\infty)=-\frac12$ to $\vp(+\infty)=\frac12$
and is determined by the relation
\begin{equation} \label{e:yvp}
    (c^2-\pi^2)(s - \vp) = \pi\tan \pi \vp \,.
    \end{equation}
\end{theorem}
The significance of the condition $c^2>\pi^2$ lies in the fact that
$\pi$ is the speed of long waves in the linearized Calogero-Moser lattice.
Thus these solitons exist with any speed exceeding the ``sound speed'' $\pi$.
These solitons are compression waves that produce a unit translation of particles 
in the direction of wave motion, with $x_j(t)$ increasing from $j-\frac12$ to $j+\frac12$
as $t$ increases from $-\infty$ to $\infty$.

The result above for solitary waves will follow by taking 
limits of waves on the infinite lattice that are periodic in space, satisfying
\begin{equation}\label{e:period}
    x_{j+N}(t) = x_j(t) + L,
\end{equation}
where $N>1$ is an integer and $L>0$ is real.  
Traveling waves of the form \eqref{e:xvp} satisfy this periodicity condition
if and only if the wave profile $\vp(s)$ satisfies 
\begin{equation}
    \vp(s+N) = \vp(s) + N-L \quad\text{for all $s$.}
\end{equation}
For such periodic waves,  since $x_{j+nN}=x_j+nL$ 
and due to the pole expansion identity 
\begin{equation}\label{e:polesCS}
    \sum_{n\in\Z} \frac2{(z-n)^3} =
    \frac{d^2}{dz^2} (\pi\cot\pi z) =  
     2\pi^3 \frac{\cos\pi z}{\sin^3 \pi z}  ,
\end{equation}
the infinite-lattice Calogero-Moser equations \eqref{e:cm1} reduce to 
Calogero-Sutherland equations for finitely many particles, namely Hamilton's equations
of motion for the Hamiltonian
\begin{equation}\label{d:HamCS}
    \mathcal{H}_{CS} = \frac12\sum_{j=1}^N p_j^2 + \frac12 \sum_{\substack{j,k=1\\j\ne k}}^N \frac{a^2}{\sin^2(a(q_j-q_k))} ,
\end{equation}
with $q_j=x_j$, $p_j=\dot x_j$, and $a=\frac\pi L$, see \cite{sutherland1971exact,Stone.Anduaga.ea.08,Philip.2019}.
Explicitly, we find the following.

\begin{theorem}[Periodic waves]\label{t:periodic}
The infinite Calogero-Moser lattice admits
wave solutions satisfying \eqref{e:xvp} and \eqref{e:period}
with $\vp(s)$ odd and monotone increasing 
being determined for $s\in(-N/2,N/2)$ by a relation of the form
\begin{equation}\label{e:pwave}
  \kappa  \tan\left(a(s-\vp)\right) = \tan\pi\vp \,.
\end{equation}
Here $a=\frac{\pi}L$ with $L=N-1$, and $\kappa>1$ is determined for any $c>\pi+a$  by 
\begin{equation}\label{e:pwave_eta}
\kappa = \frac{1+\nu}{1-\nu}, \qquad \nu = \sqrt{\frac{c^2 - (\pi+a)^2}{c^2-(\pi-a)^2}} .
\end{equation}
\end{theorem}

The proof of Theorem~\ref{t:periodic} will be provided in Section~\ref{s:periodic} below,
where we also discuss a connection to the projection method devised by Olshanetsky and Perelomov~\cite{Olshanetsky.Perelomov.81}
for the general solution of Calogero-Moser-Sutherland systems.
Theorem~\ref{t:soliton} will be derived in Section~\ref{s:soliton} through taking the limit $N\to\infty$.
Galilean transformations can be applied to these results
to obtain a broader family of waves, but we have no proof that {all} Calogero-Moser
solitary and periodic waves are obtained in this way.

The paper concludes with a discussion of how the solitary wave profiles
behave in the limits as $c\to\infty$ and as $c$ approaches $\pi$, 
along with numerical illustrations and comparison with 
wave profiles for nearest neighbor models corresponding to keeping only the
term with $m=1$ in system~\eqref{e:sys1}, especially for the case $\pow=1.73$
taken by Moler\'on~{\em et al.}~\cite{MoleronEA2014}. 

There is some evidence that the waves we find can be stable, 
as numerical computations reported by Abanov {\em et al.}~\cite{AbanovEA2011}
and Philip~\cite{Philip.2019} show localized ``1-soliton'' waves 
repeatedly passing over a finite array of particles subject to Calogero-Moser
dynamics with a weak harmonic trapping force.
The question of stability deserves a much more thorough investigation 
than we have space to undertake here, however, and we leave it for future research.

But before treating wave formulae, 
first in Section~\ref{s:longwave} we carry out 
a formal long-wave scaling analysis of the lattice equations in \eqref{e:sys1}.
When initially looking to study solitary waves on the infinite Calogero-Moser lattice
in the long-wave limit, it was surprising to us that the KdV scaling fails to be correct.
Thus it behooves us to explain what the correct scaling limit should be. 
It takes little more effort to do this for power-law forces with different exponents,
and the fact that such forces lead to the 
nonlocal continuum limits in \eqref{e:BObeta} is of general interest.

We also adapt the analysis to formally handle systems with forces alternating in sign,
as appears appropriate for modeling the experiments of \cite{MoleronEA2014}. 
Pairing consecutive terms produces an effective repulsive force that decays as $d^{-\pow-2}$ at long range. 
For $\pow>2$ this results in a KdV scaling, as one may expect from the case of purely repulsive forces. 
For $0<\pow<2$ one might expect to get a nonlocal PDE of the form \eqref{e:BObeta} with $\pow$ replaced by $\pow+1$.  
Thus it is quite surprising that instead a KdV scaling still works, 
for all $\pow>0$.

\section{The long-wave scaling limit}\label{s:longwave}

{  
The lattice equations \eqref{e:sys1} are in equilibrium 
for particles with a uniform spacing that may be taken
to be unity after a trivial scaling.  
Considering perturbations $x_j=j+\epsilon v_j$ 
about this equilibrium solution and retaining only terms  
of order $\epsilon$ results in the linearized system
\begin{equation}
   \ddot v_j =  \pow(\pow+1)\sum_{m=1}^\infty \frac{v_{j+m}-2v_j+v_{j-m}}{m^{\pow+2}}\,.
\end{equation}
Seeking solutions $v_j(t)=e^{i(kj-\omega t)}$ with wave number $k$ 
yields a dispersion relation  
with squared phase speed 
\begin{equation}
    \frac{\omega^2}{k^2} =\pow(\pow+1)\sum_{m=1}^\infty \frac{\sinc^2(\frac12 km)}{m^\pow}\,,
    \qquad \sinc x = \frac{\sin x}{x}\,.
\end{equation}
The maximal linear wave speed appears in the long-wave limit,
where we get
\begin{equation} \label{csound}
\left|\frac{\omega}{k}\right| 
\to c_\pow := \sqrt{\pow(\pow+1)\zeta_\pow}\,,
\end{equation}
in terms of  the Riemann zeta function denoted
$\zeta_s=\sum_{m=1}^\infty m^{-s}$. 
In particular 
the long-wave speed in the Calogero-Moser lattice is $c_2=\pi$,
since $\zeta_2=\frac{\pi^2}6$.

This long-wave limit 
formally leads to the expectation that 
the scaling ansatz $x_j=j+\eps v(\eps j,\eps t)$ 
should require
$v(x,\tau)$ to approximate a solution of the wave equation
\begin{equation}
\D_\tau^2 v = c_\pow^2\D_x^2 v \,,
\end{equation}
up to residual errors that vanish as $\eps\to0$ for times 
$t$ of order $O(1/\eps)$.
In traditional fashion, we now examine the effects
of dispersion and nonlinearity  
on long waves traveling in one direction 
over longer time scales,
by making the scaling ansatz 
\begin{equation} \label{e:scalepq}
    x_j = j + \eps^p v(\eps(j-c_\pow t), \eps^q t)\,.
\end{equation}
The case $p=1$, $q=3$ corresponds to the classical KdV scaling.

For the sake of clarity regarding the results of 
formal scaling analysis, let us define the {\em lattice error}
of the ansatz \eqref{e:scalepq} in equation \eqref{e:sys1}
to be the result of substituting \eqref{e:scalepq} into the expression
\begin{equation}\label{e:latticeR}
R_\eps = \ddot x_j + \pow \sum_{m=1}^\infty  
    \Bigl(  (x_{j+m}-x_j)^{-\pow-1} -(x_j-x_{j-m})^{-\pow-1}  \Bigr)\,.
\end{equation}
We consider this as a function $R_\eps=R_\eps(x,\tau)$
where $x=\eps(j-c_\pow t)$ and $\tau=\eps^q t$. 
The result of formal scaling analysis will be to show 
that for a suitably ``nice'' 
function $v(x,\tau)$, taken as {\em fixed}, 
the lattice error takes the form
\begin{equation} \label{err:lat1}
R_\eps(x,\tau) = 
\eps^{p+q+1} Q(x,\tau) + o( \eps^{p+q+1} )
\end{equation}
in the limit $\eps\to0$. The function $Q$ is independent of $\eps$
and  is the error of substituting $u=-\D_x v$ 
after a simple scaling 
into either a nonlocal PDE of the form~\eqref{e:BObeta},
or the KdV equation 
\begin{equation}\label{e:kdv}
    \D_\tau u + u \D_x u + \D_x^3u = 0.
\end{equation}
Notably, the lattice error $R_\eps$ will be $o(\eps^{p+q+1})$ 
if and only if  $Q=0$, meaning $u$ is a solution of 
the nonlocal PDE or the KdV equation in the appropriate case. 

\begin{theorem}\label{t:longwave}
Let $\pow>1$ with $\pow\ne3$. 
Assume $v(x,\tau)$ is smooth with square-integrable derivatives
$\D_x^jv$ for $1\le j\le 5$. 
Then with 
\[ u(x,\tau)=-\D_x v(x,\tau)\,,
\quad  \kappa_1=2c_\pow\,, 
\quad \kappa_2=\pow(\pow+1)(\pow+2) \zeta_\pow\,, 
\]
the lattice error relation \eqref{err:lat1} holds as follows.  
\begin{itemize}
\item[(i)] For $\pow>3$, $p=1$, $q=3$, we have
$R_\eps = \eps^5 Q+o(\eps^5)$ with
\[ Q = \kappa_1\,\D_\tau u +\kappa_2\,   u \D_x u + \kappa_3\,  \D_x^3 u ,   
\quad \kappa_3 = \tfrac1{12}\pow(\pow+1)\zeta_{\pow-2}\,.\]
\item[(ii)] For $1<\pow<3$, $p=\pow-2$, $q=\pow$,
we have $R_\eps = \eps^{2\pow-1}Q+o(\eps^{2\pow-1})$ with
\[ 
Q = \kappa_1\, \D_\tau u + \kappa_2\,  u \D_x u + \kappa_3\,  H|D|^\alpha u ,   
\quad    \kappa_3 = \pow(\pow+1)\int_0^\infty \frac{1-\sinc^2(x/2)}{x^\pow}\,dx\,.
   \]
\end{itemize}
\end{theorem}

\noindent
\begin{remark}\label{r:log}
The case $\pow=3$ requires a logarithmic correction to the KdV scaling.
In Appendix~\ref{a:logkdv}, we show that 
if \eqref{e:scalepq} is replaced in this case  by the scaling ansatz
\begin{equation}
    x_j = j + \eps\log(1/\eps) v(\eps(j-c_3 t),\eps^3\log(1/\eps) t),
\end{equation}
then $R_\eps=\eps^5\log^2(1/\eps)(Q+o(1))$ where $Q$ is as in part (i) but with $\kappa_3=1.$
\end{remark}

\begin{remark}
The PDE errors take a simpler form  after a scaling.
We find that
in case (i),
$Q = \D_\tau \tilde u + \tilde u\D_x\tilde u + \D_x^3\tilde u$, 
where 
$\tilde u(x,\tau)=\gamma^2 u(\gamma a x,\gamma^3 b\tau)$ with 
\[
\quad 
a^2 = \frac{\kappa_3}{\kappa_2}, 
\quad b = \frac{\kappa_1 a}{\kappa_2},  
\quad \gamma^{5} = \frac{\kappa_2}{a}\,.
\]
In case (ii), 
$Q = \D_\tau \tilde u + \tilde u\D_x\tilde u + H|D|^\alpha \tilde u$, where
$\tilde u(x,\tau)=\gamma^{\pow-1} u(\gamma a x,\gamma^\pow b\tau)$ with 
\[
\quad
a^{\pow-1} = \frac{\kappa_3}{\kappa_2}\,, \quad
b = \frac{\kappa_1a}{\kappa_2}, \quad
\gamma^{2\pow-1} = \frac{\kappa_2}{a} \,.
\]
\end{remark}

We emphasize that Theorem~\ref{t:longwave} is the result of a purely formal long-wave analysis.
Of course, it would be desirable to prove a long-wave approximation theorem
that compares true solutions of the lattice system~\eqref{e:sys1} to solutions of
the nonlocal PDE~\eqref{e:BObeta} over the appropriate time scale. Such an analysis
is beyond the scope of the present paper, however. We expect it would involve
delicate stability estimates for dispersive wave propagation such as have been used
to justify KdV limits in various fluid and lattice 
systems~\cite{Craig85,SchneiderWayne99,SchneiderWayne2000,HongEA2021}.

\begin{proof}
From \eqref{e:scalepq}, it is convenient to express
differences of lattice particle positions in terms of $u$
as follows. We write
\begin{align*}
 x_{j+m}-x_j &= m +\eps^p (v(x+\eps m,\tau)-v(x,\tau)) 
= m(1-\eps^{p+1} A_{\eps m}u)  ,
\\
 x_j - x_{j-m} &= m +\eps^p (v(x,\tau)-v(x-\eps m,\tau)) 
= m(1-\eps^{p+1} A_{-\eps m}u) , 
\end{align*}
in terms of the averaging operator defined for $h\ne0$ by
\begin{equation}\label{d:Ah}
A_hu(x,\tau) = \frac1h\int_0^h u(x+z,\tau)\,dz \,.
\end{equation}
By our assumptions this is uniformly bounded, 
with 
\begin{equation} \label{bound:A}
|A_h u(x,\tau)|\le \|u\|_\infty = O(1).
\end{equation}
Then with the shorthand 
$\pow_1=\pow+1$, $\pow_2=\frac12(\pow+1)(\pow+2)$,
Taylor expansion  
yields 
\begin{align*}
\frac{m^{\pow+1}}{(x_{j+m}-x_j)^{\pow+1}} &=
1 + \pow_1 \eps^{p+1} A_{\eps m} u 
+ \pow_2 \eps^{2p+2} (A_{\eps m}u)^2 +  O(\eps^{3p+3}) \,,
\\
\frac{m^{\pow+1}}{(x_j-x_{j-m})^{\pow+1}} &=
1 + \pow_1 \eps^{p+1} A_{-\eps m} u 
+ \pow_2 \eps^{2p+2} (A_{-\eps m}u)^2 + O(\eps^{3p+3}) \,.
\end{align*}
Then we can write \eqref{e:latticeR} as 
\begin{equation}\label{d:Reps}
R_\eps = \ddot x_j+ \pow\pow_1 L_\eps + \pow\pow_2 N_\eps + O(\eps^{3p+3}) \,,
\end{equation}
where the acceleration, linear and nonlinear terms are given by
\begin{align}
\label{d:ddxj}
\ddot x_j &= -\eps^{p+2}c_\pow^2\D_xu + 2 \eps^{p+q+1}c_\pow \D_\tau u + \eps^{p+2q}\D_\tau^2 v \,,
\\
\label{d:Leps}
 L_\eps &= \eps^{p+1} \sum_{m=1}^\infty    \frac1{m^{\pow+1}} 
( A_{\eps m}u - A_{-\eps m}u )
 \,,
 \\
 \label{d:Neps}
 N_\eps &= \eps^{2p+2} \sum_{m=1}^\infty    \frac1{m^{\pow+1}} 
( A_{\eps m}u + A_{-\eps m}u ) 
( A_{\eps m}u - A_{-\eps m}u ) \,.
\end{align}

Let us first estimate factors in the nonlinear term. 
\begin{lemma}\label{l:Neps}
   For fixed $x,\tau$,  we have
   $ N_\eps=  \eps^{2p+3} \bigl( 2 \zeta_\pow u \D_x u \bigr) + o(\eps^{2p+3}). $
\end{lemma}
\begin{proof}
We have
\begin{equation}
 A_{\eps m}u + A_{-\eps m}u 
 = \frac1{\eps m} \int_{-\eps m}^{\eps m} u(x+z,\tau)\,dz 
 =2 u(x,\tau) + o_m(1) \,,
\end{equation}
where the notation $o_m(1)$ denotes a generic term 
that is uniformly bounded with respect to $m$ and satisfies
$o_m(1) \to 0$ as $\eps\to 0$ for each fixed $m$.
For the difference factor, we have
\begin{align}
\frac{ A_{\eps m}u - A_{-\eps m}u}{\eps m} &= 
 \frac1{(\eps m)^2} \int_0^{\eps m} u(x+z,\tau)- u(x+z-\eps m,\tau)\,dz
\nonumber \\ &= 
 \frac1{(\eps m)^2} \int_0^{\eps m} \int_{-\eps m}^0 
 \D_x u(x+y+z,\tau)\,dy \,dz
\nonumber \\ &= \D_x u(x,\tau) +  o_m(1) \,,
\label{e:diffAm}
\end{align}
since our assumptions ensure $\D_x u$ is bounded and continuous. 
By consequence we find that as $\eps\to0$,
\begin{align}
    N_\eps &= \eps^{2p+3} 
    \sum_{m=1}^\infty \frac1{m^\pow} ( 2u\D_x u+o_m(1)) \,,
\label{eq:N}
\end{align}
and the lemma follows by dominated convergence.
\end{proof}

By \eqref{e:diffAm}, we find similarly that the leading part of the linear term is
\begin{equation}
    L_\eps = \eps^{p+2} \zeta_\pow \D_x u + o(\eps^{p+2})
\end{equation}
This cancels with the term $\eps^{p+2}c_\pow^2\D_x u$
in $\ddot x_j$ since the sound speed in the linearized lattice
satisfies $c_\pow^2 = \pow\pow_1\zeta_\pow$ from \eqref{csound}. 
The dispersive term arises at the next order in the expansion of $L_\eps$.
We consider first the easier case $\pow>3$.   

\medskip
\noindent
{\bf Case (i):}
For $\pow>3$,  standard use of Taylor's theorem yields
\begin{align}
    -\frac{A_{\eps m} u - A_{-\eps m} u}{\eps m} &= 
    \frac{v(x+\eps m,\tau)-2v(x,\tau)+v(x-\eps m,\tau)}{(\eps m)^2} 
    \nonumber \\ 
    &=  \D_x^2 v + 
    \frac{(\eps m)^2}{12} (\D_x^4 v +o_m(1)),
\end{align}
since our assumptions ensure $\D_x^4 v$ is bounded and continuous. Hence we have 
\begin{align}\label{L:kdv}
    L_\eps &= \eps^{p+2}  \sum_{m=1}^\infty \frac1{m^\alpha} \left(\D_x u + 
    \frac{\eps^2 m^2}{12} (\D_x^3 u +o_m(1)) \right)
    \nonumber \\ &=
    \eps^{p+2} \zeta_\pow\D_x u + \tfrac1{12}\eps^{p+4} \zeta_{\pow-2} \D_x^3 u + o(\eps^{p+4}),
\end{align}
by dominated convergence.
Then taking $p=1$ and $q=3$ (corresponding to the KdV scaling), the dispersive and nonlinear terms balance and
we find $R_\eps = \eps^5 Q + o(\eps^5 )$ with $Q$ as stated in the Theorem.

\medskip
\noindent
{\bf Case (ii):} For $\pow\le3$  the ordinary KdV scaling fails, due to the
divergence of the series $\sum 1/m^{\pow-2}$ appearing in \eqref{L:kdv}.
To study the linear term $L_\eps u$ we take the Fourier transform, 
defined for $u\in L^1(\R)$ (suppressing dependence on $\tau$) by 
\[
\hat u(k) = \calF u(k) = \frac1{2\pi}\int_{\R} u(x)e^{-ikx} \,dx \,.
\]
Since $\widehat{\D_x u}(k)=ik\hat u(k)$ and 
$\widehat{ A_{\eps m}u}(k) =  \hat u(k) (e^{i\eps mk}-1)/{i\eps mk}$, 
we find
\begin{align}
    \widehat{ L_\eps}(k) &= 
    \eps^{p+1} \sum_{m=1}^\infty
    \frac1{m^{\pow+1}} 
    \frac{(e^{i\eps mk/2} - e^{-i\eps mk/2})^2 }
    {(i\eps m k)^2}  {(i\eps m k)} \hat u(k)
    \nonumber
    \\ &= 
     \eps^{p+2}  {ik \hat u(k)} \sum_{m=1}^\infty 
    \frac  {\sinc^2(\eps m k/2)}
    {m^\pow}
    \nonumber
\\ &= 
     \eps^{p+2}  ik \hat u(k) \left( \zeta_\pow - 
     (\eps |k|)^\pow \sum_{m=1}^\infty 
    \frac  {1-\sinc^2(\eps m |k|/2)} {(\eps m |k|)^\pow}
     \right).
     \label{e:Lhateps}
\end{align}
The last line involves a Riemann sum approximation to a convergent integral. 
Since $1<\pow<3$, we have
\begin{equation}\label{d:etapow}
     h \sum_{m=1}^\infty 
    \frac  {1-\sinc^2(mh/2)} {(mh)^\pow} 
  \  \underset{h\to0}{\longrightarrow}\ %
    \eta_\pow := \int_0^\infty \frac{1-\sinc^2(x/2)}{x^\pow}\,dx<\infty\,.
\end{equation}
Therefore we infer that as $\eps\to0$,
\begin{equation}\label{e:Lhat2}
   \widehat{L_\eps}(k) = 
     \eps^{p+2}  \hat u(k) \left( ik\zeta_\pow + \eps^{\pow-1} 
     (-i\sgn k|k|^\pow) (\eta_\pow +o_k(1)) \right).
\end{equation}
Upon Fourier inversion we find 
\begin{equation}
    L_\eps(x,\tau) = \eps^{p+2} \zeta_\pow \D_x u + \eps^{p+\pow+1}\eta_\pow H|D|^\pow u + \eps^{p+\pow+1} E_\eps(x,\tau)\,,
\end{equation}
where $\widehat{E_\eps}(k) =  \hat u(k)  |k|^\pow o_k(1)$.
Our assumptions ensure $\hat u(k)|k|^\pow$ is integrable,
for since $\frac12(1+k^2)|k|^{2\pow}\le 1+k^8$,
by the Cauchy-Schwarz inequality we have
\begin{align*}
\left( \intR |\hat u(k)||k|^\pow \,dk \right)^2 &\le 
2\intR |\hat u(k)|^2(1+k^8)\,dk
\,\intR \frac{dk}{1+k^2} 
\\
&= C \intR u^2+ (\D_x^4 u)^2 \,dx <\infty \,.
\end{align*}
By Fourier inversion and dominated convergence
it follows $E_\eps(x,\tau)=o(1)$.
Taking $p=\pow-2$ and $q=\pow$, the linear and nonlinear terms in $R_\eps$ 
now balance and we find
$R_\eps = \eps^{2\pow-1}Q + o(\eps^{2\pow-1})$ with $Q$ as stated in the Theorem.
\end{proof}

\begin{remark}\label{r:eta}
An explicit formula for the integral $\eta_\pow$ in \eqref{d:etapow} is
\begin{equation}\label{e:rmt_result}
 \eta_\pow= 
 \begin{cases}
-2 \sin\left(\pi\alpha/2\right) \Gamma(-1-\alpha) \,,
  & \pow\in(1,2)\cup(2,3),\\
  \ \ \pi/6\,, & \pow=2.
 \end{cases}
\end{equation}
This formula for general $\pow$ is motivated from the form of Ramanujan's Master Theorem \cite{Tew2012}, 
which relates to Mellin transforms. We were not able to verify the hypotheses of this theorem,
unfortunately, but instead found the following rather uncomplicated direct proof
of \eqref{e:rmt_result}: Note
    $\eta_\pow = 2\int_0^\infty x^{2-\pow} f(x)\,dx$ where 
\begin{equation}\label{d:ffrac}
    f(x):= \frac{1-\sinc^2(\frac12 x)}{2x^2} = 
    \frac{\cos x - 1 + \frac12 x^2}{x^4}.
\end{equation}
For $s,p>0$ we have $
\int_0^\infty x^{s-1} e^{-px}\,dx = p^{-s}\,\Gamma(s)$,
and this formula extends by analytic continuation to hold whenever $\re s$ and $\re p>0$.
Taking $p=\eps\pm i$ with $\eps>0$ we find
\begin{align*}
    I(s,\eps)&:= \int_0^\infty x^{s-1}x^4f(x)e^{-\eps x}\,dx
    \\ &\,= \frac12\left( (\eps-i)^{-s} +(\eps+i)^{-s} \right)\Gamma(s) 
    -\eps^{-s}\Gamma(s) + \frac12\eps^{-s-2}\Gamma(s+2) \,.
\end{align*}
The integral $I(s,\eps)$ is analytic in the half-plane where $\re s>-4$, 
and this formula extends analytically to this half-plane. 
For $\re s\in(-4,-2)$ with $s\ne-3$ we can take the limit $\eps\downarrow0$ 
and infer 
\begin{equation}\label{e:Is0}
I(s,0) = \cos\left(\frac{\pi s}2\right)\Gamma(s)\,.
\end{equation}
Taking $s=-1-\pow$ we deduce the first formula in \eqref{e:rmt_result}.
To get the second formula, take $s\to-3$.
\end{remark}

\subsection*{Alternating forces}
  For a system having power-law interaction forces that alternately repel and attract,
  given by 
\begin{equation}\label{e:sysalt}
    \ddot x_j = 
    -\pow \sum_{m=1}^\infty  
    \Bigl(  (x_{j+m}-x_j)^{-\pow-1} -(x_j-x_{j-m})^{-\pow-1}  \Bigr)(-1)^{m-1}\,,
\end{equation}
  we find that the KdV scaling as in part (i) of Theorem~\ref{t:longwave} 
  works {\em for all $\pow>0$.} 
  The only change in the statement, aside from including a factor $(-1)^m$
  in the definition of the lattice error $R_\eps$ in \eqref{e:latticeR},
  is that for determining the sound speed $c_\pow$
  and the coefficients, the zeta function values $\zeta_\pow$ and $\zeta_{\pow-2}$
  should be respectively replaced by values $\zeta_\pow^*$ and $\zeta_{\pow-2}^*$ of the 
  {\em alternating zeta function} given by $ \zeta_s^* = \zeta_s(1-2^{1-s})$,
  satisfying $\zeta_s^* = \sum_{m=1}^\infty (-1)^{m-1}m^{-s}$ for $\re s>0$.

  \begin{theorem}\label{t:alt}
  Let $\pow>0$ and take $v$ and $u$ as in Theorem~\ref{t:longwave}. 
  Then under the ansatz~\eqref{e:scalepq} 
  with $p=1$, $q=3$, $c_\pow=\sqrt{\pow(\pow+1)\zeta_\pow^*}$,
  the lattice error $R_\eps$ for system \eqref{e:sysalt} 
  satisfies $R_\eps=\eps^5 Q+o(\eps^5)$, where 
\[
Q = \kappa_1\,\D_\tau u +\kappa_2\,   u \D_x u + \kappa_3\,  \D_x^3 u , 
\]
with $\kappa_1 = 2c_\pow$,
$\kappa_2=\pow(\pow+1)(\pow+2)\zeta_\pow^*$, and
$\kappa_3 = \frac1{12}\pow(\pow+1)\zeta_{\pow-2}^*$.
  \end{theorem}
  
  When $\pow>3$ the proof is a simple modification
  of the arguments above for proving part (i) of Theorem~\ref{t:longwave}. 
  For $0<\pow\le 3$ the proof is a modification of the proof of part (ii), with the 
  only essential change being that the expression in \eqref{e:Lhateps} now takes the form
  \begin{align}\label{e:Lhatalt}
  \widehat{L_\eps}(k) &= \eps^3 ik\hat u(k)\Bigl( \zeta_\pow^* - S_\pow(\eps|k|)\Bigr),
  \\ \label{d:Spow}
  S_\pow(h) &= \sum_{m=1}^\infty \frac{1-\sinc^2(mh/2)}{m^\pow}(-1)^{m-1}.
  \end{align}
  Then based on the following lemma, one finds that \eqref{e:Lhat2} changes to
  \begin{equation}
  \widehat L_\eps(k) = \eps^3\zeta_\pow^*\widehat{\D_x u}(k) + 
  \tfrac1{12}\eps^5\zeta_{\pow-2}^* \widehat{\D_x^3 u}(k) (1+o_k(1)) ,
  \end{equation}
  and the rest of the proof goes as before. 
  \begin{lemma}\label{l:Spow} For any $\pow>0$, we have 
  $S_\pow(h) = \tfrac1{12}\zeta_{\pow-2}^*\, h^2 + O(h^3)$ as $h\to0$.
  \end{lemma}
  We prove this lemma in Appendix~\ref{a:alt} using 
  the inversion formula for the Mellin transform and path deformation;
  see~\cite{Flajolet} for this method.
} 

\section{Periodic Calogero-Moser-Sutherland waves}\label{s:periodic}
\subsection{B\"acklund transforms for Calogero-Sutherland systems}
Our strategy to prove Theorem~\ref{t:periodic} 
involves equations for Calogero-Sutherland systems introduced by 
Wojciechowski \cite{Wojciechowski.82} that
he called an analogue of the B\"acklund transformations known for other integrable systems.
The equations couple $N$ particle positions $x_1,\ldots,x_N\in\C$ 
with $M$ ``dual'' particle positions $y_1,\ldots,y_M\in\C$. 
In the case we will use, they take the form
\begin{align}
i\dot x_j &=  
\sum_{\substack{k=1\\k\ne j}}^N 
a\cot a(x_j-x_k) 
 - \sum_{m=1}^M a\cot a(x_j-y_m) \,,
 \label{e:perback1}
\\
i\dot y_n &= \sum_{k=1}^N a\cot a(y_n-x_k)   
  - \sum_{\substack{m=1\\m\ne n}}^M a\cot a(y_n-y_m) \,.
 \label{e:perback2}
\end{align}
For any solution of these coupled equations,
it is well known (but see Appendix~\ref{a:backCS} for an efficient proof)
that $x_1,\ldots,x_N$ and $y_1,\ldots,y_M$ separately solve decoupled Calogero-Sutherland systems,
with 
\begin{align}
\label{e:CSx}
\ddot x_j &= 2a^3 \sum_{\substack{k=1\\k\ne j}}^N 
\cos a(x_j-x_k) \sin^{-3}a(x_j-x_k) \,,
\\
\label{e:CSy}
\ddot y_n &= 2a^3 \sum_{\substack{m=1\\m\ne n}}^M 
\cos a(y_n-y_m) \sin^{-3}a(y_n-y_m) \,.
\end{align}

Several authors \cite{AbanovEA2011,Stone.Anduaga.ea.08,Philip.2019} 
refer to solutions of the coupled system \eqref{e:perback1}--\eqref{e:perback2}
as providing ``$M$-soliton'' solutions of the Calogero-Moser-Sutherland system \eqref{e:CSx}.
Possibly this terminology is motivated by the connection, through pole dynamics,
with rational $N$-soliton solutions of the 
Benjamin-Ono equations in the case when $N=M$ and $y_j=\bar x_j$ 
and when the function $\phi(r)=a\cot ar$ is replaced by $\phi(r)=1/r$ above~\cite{Case.78}.
In this rational case when $\phi(r)=1/r$ a harmonic force term is sometimes included. 

\subsection{Steps to prove Theorem~\ref{t:periodic}}
Throughout this section we assume $a=\frac\pi{L}$ with $L=N-1\in\N$.
The proof of Theorem~\ref{t:periodic} will have four main steps:
\begin{enumerate}
\item First, we show that for any $\kappa>1$, the relation \eqref{e:pwave},
together with the periodicity property
\begin{equation}\label{e:phiN}
\vp(s+N)=\vp(s)+1 \,,
\end{equation}
determines a unique strictly increasing real analytic function $\vp(s)$
on the line, and that \eqref{e:xvp} then defines lattice particle positions 
$x_j(t)$ for $j\in\Z$ with the desired periodic wave symmetries 
in \eqref{e:wavesymmetry} and \eqref{e:period}.

\item Next, we infer that corresponding points on the unit circle in $\C$, given by 
\begin{equation}\label{d:zj}
 z_j = e^{2ia x_j}  \,, \quad j=1,\ldots,N,
\end{equation}
comprise $N$ distinct roots of a certain polynomial of degree $N$, given by 
\begin{equation}\label{d:polyP}
   P(z;\ict):= z^N - \nu\ict z^{N-1} + \nu z - \ict\,, 
   \quad \nu= \frac{\kappa-1}{\kappa+1}, \quad \ict = e^{2iact} .
\end{equation}
\item Third, under the assumption that $\nu$ and $c$
are related as in \eqref{e:pwave_eta}, we deduce that 
the B\"acklund transform equations \eqref{e:perback1}--\eqref{e:perback2} hold, 
with $M=1$ and with $y_1(t)=y_0(t)+L$ where 
\begin{equation}
y_0(t) = ct - i\bshift,
\end{equation}
for a certain value of $\bshift$ determined by $c$ and $N$.
\item 
The final step is simply to deduce that thus
$x_1,\ldots,x_N$ satisfy the Calogero-Sutherland equations~\eqref{e:CSx},
and therefore the $x_j$ (for $j\in\Z$) form an $N$-periodic wave solution
of the Calogero-Moser system~\eqref{e:cm1}.
\end{enumerate}

We remark that our discovery of the determining formula 
\eqref{e:pwave} for the wave profile proceeded by seeking traveling-wave
solutions for the B\"acklund transform equations, and ignoring the real part of \eqref{e:perback1}.
We omit this heuristic derivation as it is somewhat involved and would muddy the logic of the 
rigorous proof.

\subsection{Profile and wave symmetries}
\begin{lemma}
Let $\kappa>1$ be arbitrary. Then there is a unique strictly increasing real analytic function 
$\vp\colon\R\to\R$ such that the relation~\eqref{e:pwave} holds, 
together with the periodic-shift condition \eqref{e:phiN}.
\end{lemma}
\begin{proof}
We first determine $y$ (later $=(s-\vp)/L$) as a function of $\vp$ so that 
\begin{equation}\label{e:vp_y}
\kappa\tan \pi y = \tan\pi\vp  \quad\text{and}\quad
\kappa\cot\pi\vp = \cot\pi y.
\end{equation}
One checks that $y=\hat y(\vp)$ can be defined on $\R$ by direct integration from 
\begin{equation}\label{e:yintvp}
y = \int \frac{\kappa\,d\vp}{\kappa^2\cos^2 \pi \vp + \sin^2\pi \vp}, \qquad y(0)=0,
\end{equation}
after substituting $\kappa w=\tan\pi\vp$.
Clearly $\hat y\colon\R\to\R$ is odd, strictly increasing, surjective and real analytic, 
and moreover $\hat y(\vp+1)=\hat y(\vp)+1$ for all $\vp$.

Next, with $s=\hat s(\vp)$ defined by $s= \vp + L y$, clearly relation~\eqref{e:pwave}
holds, and moreover $\hat s$ is odd, strictly increasing, surjective and real analytic, 
with $\hat s(\vp+1)=\hat s(\vp)+N$ for all $\vp$. Upon inverting, we find $\vp$ as a function 
of $s$ with all the properties claimed.
\end{proof}

\begin{corollary}\label{cor:xj}
Let $c>0$.  With $x_j(t)$ given by~\eqref{e:xvp} for all $j\in\Z$, 
the traveling-wave symmetry condition~\eqref{e:wavesymmetry} 
(with $h=1=c\tau$)
and the periodic-wave symmetry condition~\eqref{e:period} both hold.
Moreover, for all $j\in\Z$, 
\[
x_j<x_{j+1}<\ldots<x_{j+N}=x_j+L. 
\]
\end{corollary}
\begin{proof} 
Observe $x_j-ct = j+s - \vp(j+s)$ where $s=-ct$.  
The symmetry~\eqref{e:wavesymmetry} is easy to check.  
Also, $x_j$ is strictly increasing in $j$ since $s-\vp(s) =L \hat y(\vp(s))$  is strictly increasing in $s$.
And, \eqref{e:phiN} implies $x_{j+N}=x_j+L$, hence the result.  
\end{proof}

\subsection{Periodic waves and polynomial roots}
\begin{lemma}\label{l:roots}
Let $z_j=e^{2ia x_j}$ for all $j\in\Z$. Then for all real $t$,
the values
$z_1(t),\ldots,z_N(t)$ are distinct and comprise all $N$ roots of the polynomial 
\[
   P(z;\ict):= z^N - \nu\ict z^{N-1} + \nu z - \ict\,, 
   \quad\text{with}\quad \nu= \frac{\kappa-1}{\kappa+1}, \quad \ict = e^{2iact} .
\]
\end{lemma}
\begin{proof}
It follows from Corollary~\ref{cor:xj} that $z_{j+N}=z_j$ for all $j$ and that
$z_1,\ldots,z_N$ are distinct complex numbers on the unit circle.
Next, observe that relation \eqref{e:pwave} says that for all $s$,
\begin{equation}
\kappa\tan a(s-\vp) =    
    \frac{\kappa}i \, \frac{e^{2ias}-e^{2ia\vp}}{e^{2ias}+e^{2ia\vp}}
    = \frac{1}i \, \frac{e^{2i\pi\vp}-1}{e^{2i\pi\vp}+1} 
= \tan \pi\vp \,.
\end{equation}
In terms of $u=e^{2ia\vp(s)}=e^{2i\pi\vp/L}$ and noting $\mict=e^{2ias}$, this is equivalent to 
\[
0 = \kappa(u-\mict)(u^L+1) +  (u+\mict)(u^L-1),
\]
and again to the polynomial equation 
\begin{equation}\label{e:polyu}
    0 = u^{L+1} -\nu\mict u^L+\nu u-\mict = P(u;\mict) = P(e^{2ia\vp(s)};e^{2ias}).
\end{equation}
Since $\bar z_k = e^{2ia(\vp(s+k)-k)}$ and $e^{2iaL}=e^{2\pi i}=1$,
we find 
\begin{align*}
    P(\bar z_k;\mict) &= 
    P(e^{2ia(\vp(s+k)-k)};e^{2ias}) 
    \\ &= 
    e^{-2iak}  P(e^{2ia\vp(s+k)}; e^{2ia(s+k)}) = 0.
\end{align*}
Upon conjugation we obtain $P(z_k;\ict)=0$, for every $k\in\Z$ and $t\in\R$.
\end{proof}

\subsection{Validity of B\"acklund transform equations}
\begin{proposition}\label{p:backlund}
Let $c>\pi+a$,  let $y_0(t)=ct-i\bshift$, and assume $\mu=e^{2a\bshift}$ 
and 
\begin{equation}\label{d:mu}
\mu = \sqrt{ \frac{(c+a)^2-\pi^2}{(c-a)^2-\pi^2} }, \qquad 
\nu = \sqrt{\frac{c^2 - (\pi+a)^2}{c^2-(\pi-a)^2}} .
\end{equation}
Then the following B\"acklund transform equations hold:
\begin{align}
 \label{e:backCS1}
    i\dot x_j &=  
    \sum_{\substack{k=1\\k\ne j}}^N
    a\cot(a(x_j-x_k)) - a\cot(a(x_j-y_0))
    \,, \\
 \label{e:backCS2}
    i\dot y_0 &= 
    \sum_{k=1}^N a\cot(a(y_0-x_k)) \,.
\end{align}
\end{proposition}

Before starting the proof proper, we express equation \eqref{e:backCS1}
in terms of the variables $z_j$ using the identities
\begin{equation}\label{e:cot_ident}
\cot(x-y)=i \frac{e^{2ix}+e^{2iy}}{e^{2ix}-e^{2iy}} \,,\qquad
e^{2iay_0}=e^{2iact}e^{2a\bshift}=\sigma\mu.
\end{equation}
Then \eqref{e:backCS1} is equivalent to 
\begin{align}
\label{e:back_zj1}
\frac1{2ia}\frac{\dot z_j}{z_j}  &= 
    a \sum_{\substack{k=1\\k\ne j}}^N \frac{z_j+z_k}{z_j-z_k} 
    - a\frac{z_j+\ict\mu}{z_j-\ict\mu} \,.
\end{align}
The sum can be expressed in terms of $z_j$ alone, in terms of $P(z)=P(z;\ict)$:

\begin{lemma} 
For each $j$ it holds that
\begin{align*}
    \sum_{\substack{k=1\\k\ne j}}^N \frac{z_j+z_k}{z_j-z_k} 
    &= \frac{z_j P''(z_j)-LP'(z_j)}{P'(z_j)}
    = \frac{L\nu\ict z_j^{N-2}-L\nu}{(L+1) z_j^{N-1} - L\nu\ict z_j^{N-2}+\nu} \,.
\end{align*}
\end{lemma}
\begin{proof}
    Fix $j$ and note that
    \[
    \sum_{\substack{k=1\\k\ne j}}^N \frac{
    z_j+z_k}{z_j-z_k} 
    = 
    \sum_{\substack{k=1\\k\ne j}}^N \frac{
    z_k-z_j+2z_j}{z_j-z_k} 
    =
    -L + 2z_j 
    \sum_{\substack{k=1\\k\ne j}}^N \frac{1}{z_j-z_k} \,.
    \]
    Now, since $P(z)=\prod_{k=1}^N(z-z_k)$, for $P(z)\ne0$ we have
    \[
    \sum_{\substack{k=1\\k\ne j}}^N \frac{1}{z-z_k} 
    = \frac{P'(z)}{P(z)} - \frac1{z-z_j} = \frac{P'(z)(z-z_j)-P(z)}{P(z)(z-z_j)}.
    \]
    Then by Taylor expansion at $z_j$ (or L'H\^opital's rule), taking $z\to z_j$ yields
\[
    \sum_{\substack{k=1\\k\ne j}}^N \frac{1}{z_j-z_k}  =\frac{P''(z_j)}{2P'(z_j)}. 
\]
This proves the Lemma.
\end{proof}

\begin{proof}[Proof of Proposition~\ref{p:backlund}]
1. We will prove~\eqref{e:back_zj1} first. 
To begin we note that $\mu$ and $\nu$ are related by the equations
\begin{equation}\label{e:munu}
\mu\nu = \frac{c+a+\pi}{c-a+\pi},
\qquad \frac\mu\nu = \frac{c+a-\pi}{c-a-\pi}.
\end{equation}
Next, differentiation of $P(z_j;\ict)=0$ yields, since $\dot\ict=2iac\ict$, 
\[ 
\dot z_jP'(z_j)=\dot\ict(\nu z_j^{N-1}+1) = 2iac(z_j^N +\nu z_j)  .
\]
Combining this with the last Lemma, in order to verify \eqref{e:back_zj1} it suffices to show
\begin{equation}\label{e:back3}
c \frac{z^{N-1}+\nu}{P'(z)} = aL\nu \frac{\ict z^{N-2}-1}{P'(z)} 
    - a\frac{z+\ict\mu}{z-\ict\mu} \,.
\end{equation}
We drop the subscript, writing $z=z_j$ here and for the rest of this step of the proof. 
Cross-multiplying, we find \eqref{e:back3} is equivalent to 
\begin{align*}
0 &= (cz^{N-1}+c\nu -aL\nu\ict z^{N-2} + aL\nu)(z-\ict\mu)
\\& \quad + a( (L+1)z^{N-1}-L\nu\ict z^{N-2} + \nu)(z+\ict\mu) \,.
\end{align*}
Because $aL=\pi$ we find this equivalent to 
\begin{align*}
0 &= (c+\pi+a)(z^N+\nu z)
-\ict z^{N-1}((c-\pi-a)\mu+2\pi\nu)
-\sigma\mu\nu(c+\pi-a).
\end{align*}
But due to the relations \eqref{e:munu} this is equivalent to
\[
0= (c+\pi+a)P(z;\ict) \,,
\]
which is true. This completes the proof of~\eqref{e:back_zj1}.

2. It remains to prove~\eqref{e:backCS2}. Note that by summing \eqref{e:backCS1}
we find 
\[
\sum_{k=1}^Na\cot a(y_0-x_k) = \sum_{k=1}^N i\dot x_k - 
\sum_{k=1}^N\sum_{\substack{l=1\\l\ne k}}^N a \cot a(x_k-x_l). 
\]
The double sum vanishes since terms cancel in pairs upon switching $k$ and $l$.
Thus to prove \eqref{e:backCS2} it suffices to show
\begin{equation}\label{e:c_dotxk}
c = \sum_{k=1}^N \dot x_k \,.
\end{equation}
But since $P(z) = \prod_{k=1}^N (z-z_k)$ we get 
$P(0)=-\ict = (-1)^N \prod_{k=1}^N z_k$,
so that 
\[
\ict = e^{2iact}
= (-1)^L \exp\left(2ia\sum_{k=1}^N x_k\right) . 
\]
Upon differentiating this, we infer \eqref{e:c_dotxk} is valid, and this proves \eqref{e:backCS2}.
\end{proof}

\subsection{Conclusion of the proof of Theorem~\ref{t:periodic}}

From the B\"acklund equations in Proposition~\ref{p:backlund}, it follows 
that the Calogero-Sutherland equations \eqref{e:CSx} hold. 
This implies, due to the pole expansion identity 
\eqref{e:polesCS},
that the infinite sequence $x_j(t)$, $j\in\Z$, which satisfies $x_{j+nN}=x_j+nL$ for all $j$ and $n$,
satisfies the Calogero-Moser system \eqref{e:cm1}, for \eqref{e:CSx} and \eqref{e:polesCS} imply  
\[
\ddot x_j = 
\frac 2{L^3} \sum_{\substack{k=1\\k\ne j}}^N \sum_{n\in\N} \left(\frac{x_j-x_k}L - n\right)^{-3}
= \sum_{\substack{k\in\Z\\k\ne j}} \frac{2}{(x_j-x_k)^3}.
\]
Note the terms in the last sum with $k=j+nN$ cancel in opposite-sign pairs.

\subsection{Relation to the projection method}

The solution of the general initial-value problem for the Calogero-Sutherland system
can be described by means of the so-called {\em projection method} of
Olshanetsky and Perelomov~\cite{Olshanetsky.Perelomov.81}.
We have not made any use of the projection method in deriving or verifying 
the formulas for periodic waves in Theorem~\ref{t:periodic}. 
But there appears to be a relation to it which we can only partially explain, 
going through the polynomial root properties from Lemma~\ref{l:roots}. 

Operationally, the projection method determines solutions as follows:
The time-dependent quantities $z_j= e^{2iax_j}$ are the eigenvalues of a matrix
expressed as 
\begin{equation}
    X(t) = B e^{2ia At} B \,,
\end{equation}
where the matrices $A$ and $B$ are explicitly given by initial data,
with entries
\begin{align}
  A_{jk} &= \delta_{jk}\, \dot x_j(0) + 
  (1-\delta_{jk})
  \frac{ia}{\sin a(x_j(0)-x_k(0))} 
  \,,
  \\  B_{jk} &=  \delta_{jk}\, e^{iax_j(0)} \,.
\end{align}
(See Appendix~\ref{a:projection_method} for an explanation, and a modified solution procedure.)

For initial data that correspond to a periodic wave given by Theorem~\ref{t:periodic},
by Lemma~\ref{l:roots} it follows that  the characteristic polynomial of $X(t)$
must be identical to the polynomial $P(z)=P(z;\ict)$, i.e.,
\begin{equation}
    \det(zI-X(t)) = P(z)   \,.
\end{equation}
Why the characteristic polynomial should have such a simple expression in this case
may be an interesting issue for further investigation.

\section{Calogero-Moser solitary waves}\label{s:soliton}

\subsection{Proof of Theorem~\ref{t:soliton}}
We now turn to the proof of Theorem~\ref{t:soliton}. Fix $c>\pi$. 
The aim is to show that 
if $\vp(s)$ is determined by \eqref{e:yvp} and $x_j(t)$ by \eqref{e:xvp},
then the Calogero-Moser equations \eqref{e:cm1} hold. It suffices to do this for $j=0$ only,
due to the fact that the shift symmetry \eqref{e:wavesymmetry} with $h=1$, $\tau=1/c$ implies
for all $j,k$ and all $t$, 
\[
x_k(t)=x_{k+j}(t+j\tau)-j.
\]

We introduce the notation 
\[
x^N_j(t)= j-\vp_N(j-ct)
\]
to denote the $N$-periodic wave solutions
of the Calogero-Moser system as described by Theorem~\ref{t:periodic}, 
where $\vp_N$ is determined by \eqref{e:pwave}.
In order to prove Theorem~\ref{t:soliton}, it suffices to prove the following three limit identities,
for every $t\in\R$:
\begin{align} \label{l1_sol}
    x_j(t) &= \lim_{N\to\infty} x^N_j(t) \,, \quad\text{for all $j\in\Z$,}
  \\ \label{l2_sol}
    \ddot x_0(t) &= \lim_{N\to\infty} \ddot x_0^N(t)  \,,
  \\ \label{l3_sol}
    \sum_{k\ne 0} \frac{2}{(x_0-x_k)^3}
    &= \lim_{N\to\infty} \sum_{k\ne 0} \frac{2}{(x^N_0-x^N_k)^3} \,.
\end{align}
To proceed, we first study the coefficients $\nu$ and $\kappa$ determined 
from $N$ by \eqref{e:pwave_eta}:
\begin{lemma}\label{lem:kappaa_lim} 
As $N\to\infty$, we have $\nu\to1$ and $\kappa a\pi \to c^2-\pi^2$.
\end{lemma}
\begin{proof}
    Recalling $a=\frac\pi L$ with $L=N-1$, this last follows from the relation
    \[
    \kappa = \frac{(1+\nu)^2}{1-\nu^2} = (1+\nu)^2 \left(\frac{c^2-(\pi-a)^2}{4\pi a}\right) \,.
    \qedhere 
    \]
\end{proof}
Evidently, both \eqref{l1_sol} and \eqref{l2_sol} follow immediately from pointwise convergence of 
the derivatives $\vp^{(n)}_N$ of $\vp_N$ to those of $\vp$:
\begin{lemma}\label{lem:limvp}
   $\vp^{(n)}(s)=\lim_{N\to\infty}\vp_N^{(n)}(s)$ for $n=0,1,2$ and all $s\in\R$.
\end{lemma}
\begin{proof} By differentiating the relations \eqref{e:yvp} and \eqref{e:pwave} that 
respectively determine $\vp$ and $\vp_N$, after a bit of calculation we find
\begin{align}
   \vp'(s) &= \frac{(c^2-\pi^2)\cos^2\pi\vp}{\pi^2+ (c^2-\pi^2)\cos^2\pi\vp},
   \\
   \vp'_N(s) &= \frac{\kappa a\pi \cos^2\pi\vp_N}{\pi^2\cos^2 a(s-\vp_N)+\kappa a\pi \cos^2\pi\vp_N}.
\end{align}
Since $\vp(0)=0=\vp_N(0)$, the pointwise convergence $\vp_N(s)\to\vp(s)$ as $N\to\infty$
(uniformly on compact sets, in fact) follows from Lemma~\ref{lem:kappaa_lim} 
by continuous dependence for initial-value problems for ODEs. 
Then $\vp'_N(s)\to\vp'(s)$ follows from the ODEs,
and $\vp''_N(s)\to\vp''(s)$ follows by differentiating the ODEs.
\end{proof}

To justify the last limit formula \eqref{l3_sol}, observe that for all $k\ne0$,
\[
|x_0(t)-x_k(t)| = |k| \left|1 -\frac{\vp(s+k)-\vp(s)}k\right|.
\]
Then from the lemma below, we obtain the bounds
\[
|x^N_0-x^N_k|\ge \delta |k|, \qquad \frac{2}{|x^N_0-x^N_k|^3}\le \frac{2}{|\delta k|^3},
\]
for $N$ sufficiently large, 
whence the limit \eqref{l3_sol} follows by dominated convergence.

\begin{lemma}
   There exists $N_0$ and $\delta>0$ such that 
   \[
   \vp_N'(s)\le 1-\delta  \quad\text{for all $s\in\R$ and $N\ge N_0$.}
   \]
\end{lemma}
\begin{proof}
Using that $s=\vp_N + L y$ with $y$ given by \eqref{e:yintvp}, 
differentiating we find that
\[
1 = \vp_N'\left( 1 + \frac{L\kappa}{\kappa^2\cos^2\pi\vp_N+\sin^2\pi\vp_N}\right) 
\ge \vp_N'\left(1+ \frac{L\kappa}{\kappa^2+1}\right)
\]
for all $s$. But by Lemma~\ref{lem:kappaa_lim}, as $N\to\infty$ we have $\kappa\to\infty$
and
\[
1+ \frac{L\kappa}{\kappa^2+1}
= 1+ \frac{\pi^2}{\kappa a\pi} \frac{\kappa^2}{\kappa^2+1}\to \frac{c^2}{c^2-\pi^2}.
\]
Hence for $N_0$ large enough, the claimed result follows with any $\delta<\pi^2/c^2$.
\end{proof}

This finishes the proof of Theorem~\ref{t:soliton}.


\subsection{Distinguished limits}

\begin{figure}
    \centering
    \includegraphics[width=0.9\textwidth]{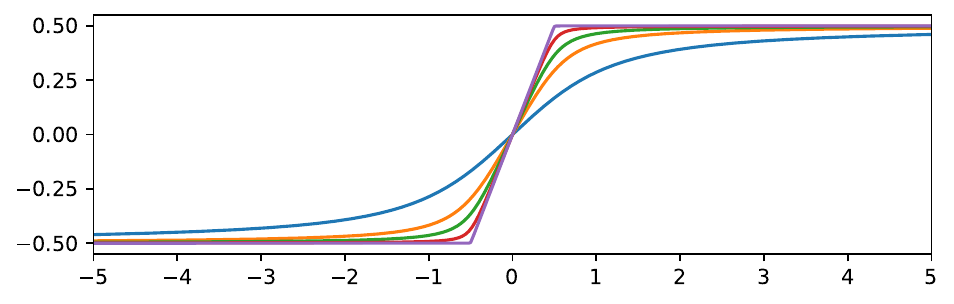}
    \put(-315,85){$\vp$}
    \\ \includegraphics[width=0.9\textwidth]{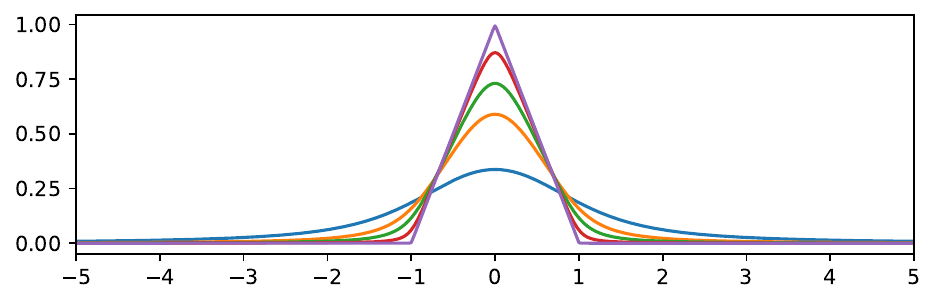}
    \put(-325,85){$-r$}
    \put(0,20){$s$}
    \caption{Profiles for soliton displacement (top) and relative displacement (bottom)
    for $c/\pi=1.25,1.75, 2.5,5,100$} 
    \label{fig:profile}
\end{figure}

In Fig.~\ref{fig:profile}, for several values of $c/\pi$, 
we plot profiles for soliton displacement $\vp(s)$ and relative displacement $-r(s) = \vp(s+1)-\vp(s)$.
As the figures suggest, the soliton formula \eqref{e:yvp} simplifies as
$c\to\infty$ and $c\to\pi$ in interesting ways. 

In the limit $c\to\infty$, evidently the profile $\vp\to\vp_\infty$,
 where $\vp_\infty$ is odd with
\begin{equation}
    \vp_\infty(s) = \min\left(s,\tfrac12\right) \quad\text{for $s\ge0$.}
\end{equation}
Thus high-speed waves converge to a {\em hard-collision limit}, in which
one particle at a time moves at a constant speed, coming to a stop when
it collides with the next particle in front.

In the (sonic) limit $c\to\pi$, if we scale by writing $c^2 = \pi^2+\eps \pi$
then we find that \eqref{e:yvp} reduces to 
\begin{equation}
\eps s = \tan\pi\vp + O(\eps).
\end{equation}
We find this consistent with the formal long-wave limit obtained in 
Theorem~\ref{t:longwave}. This limit is a Benjamin-Ono equation, which 
for $w=\pi u=-\pi\D_x v(x,\tau)$ takes the form
\begin{equation}\label{e:BO_w}
    2\D_\tau w + 4 w\D_x w - H\D_x^2 w =0  \,,
\end{equation}
since for $\alpha=2$ the coefficients 
$\kappa_1=2\pi$, $\kappa_2=4\pi^2$ and $\kappa_3=\pi$.
Equation \eqref{e:BO_w} has a solitary wave solution 
    $w(x,\tau)=W(x-\frac12\tau)$ with
\begin{equation}
  W(z) =  \re\left(\frac{i}{z+i}\right) = \frac1{z^2+1} , \quad
  HW(z) =  \im\left(\frac{i}{z+i}\right) = \frac{z}{z^2+1}  \,,
\end{equation}
which satisfy $\D_z(HW)=2W^2-W$.
Since $c=\sqrt{\pi^2+\eps\pi}\sim \pi + \frac12\eps$, 
the correspondence $z=x-\frac12\tau = \eps(j-\pi t)-\frac12\eps^2t$
is consistent with $z\sim \eps s = \eps(j-ct)$ and 
\[
W \sim \frac{\pi}\eps \frac{d\vp}{ds} \sim \frac1{(\eps s)^2+1}.
\]

\subsection{Numerical comparison with nearest-neighbor models}

\begin{figure}
    \centering
    \includegraphics[width=0.9\textwidth]{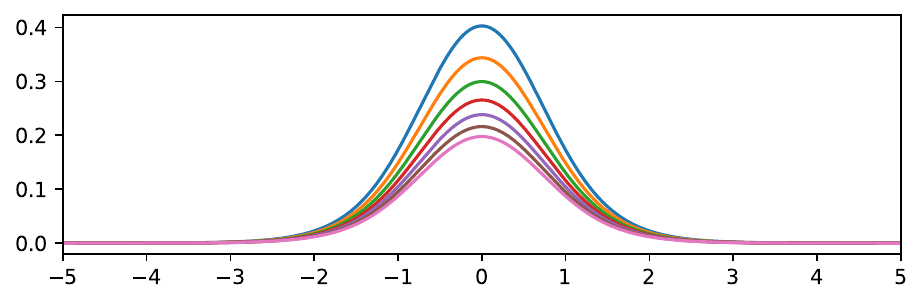}
    \put(-320,85){$-r$}
    \\ \includegraphics[width=0.9\textwidth]{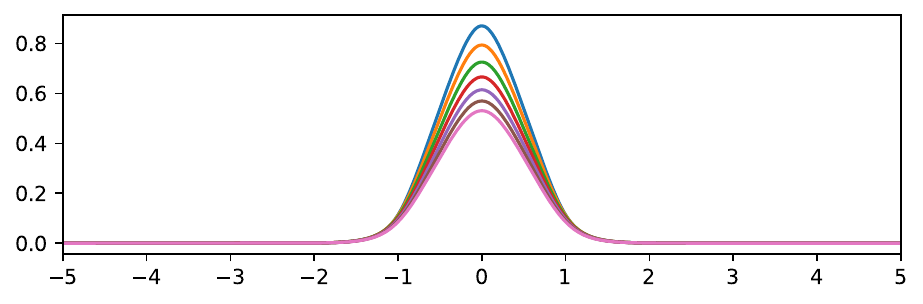}
    \put(-325,85){$-r$}
    \put(0,20){$s$}
    \caption{Relative displacement for solitary waves in 
    nearest-neighbor lattices, varying $\pow$ from $0.5$ to $3.5$ 
    (top to bottom in each subplot). $c/c_s=1.25$ (top), $2.5$ (bottom)}
    \label{fig:vary_alpha}
\end{figure}
\begin{figure}
    \centering
    \includegraphics[width=0.9\textwidth]{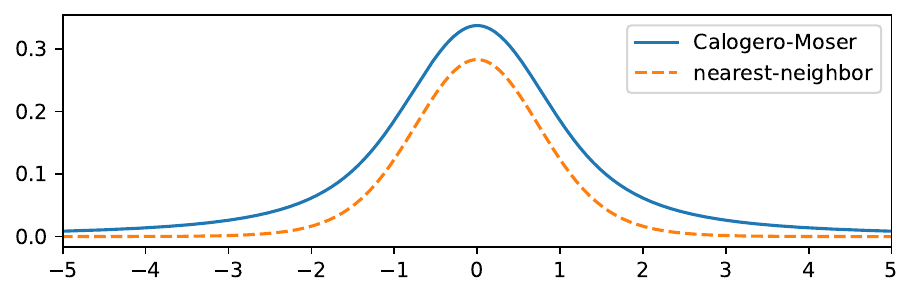}
    \put(-320,85){$-r$}
    \\ \includegraphics[width=0.9\textwidth]{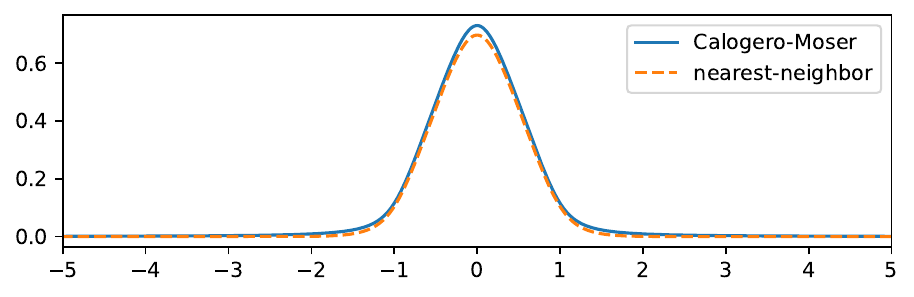}
    \put(-325,85){$-r$}
    \put(0,20){$s$}
    \caption{Relative displacement for solitary waves, comparing 
    infinite-lattice Calogero-Moser 
    to nearest-neighbor lattice with $\beta=\pow+1=2.73$  
    for fixed speed ratio $c/c_s=1.25$ (top), $2.5$ (bottom)}
    \label{fig:compare}
\end{figure}

Let us now compare relative displacement profiles for
solitary waves in the infinite Calogero-Moser lattice
with numerical computations for the power-law nearest-neighbor lattice.
Particle positions in the latter are governed by the system
\begin{equation}\label{e:meq1sys}
    \ddot x_j =  -\pow \left( (x_{j+1}-x_j)^{-\pow-1} - (x_j-x_{j-1})^{-\pow-1} \right),
\end{equation}
keeping only the term with $m=1$ on the right-hand side of system~\eqref{e:sys1}.
For solitary waves $x_j(t)=j-\vp(j-ct)$, one can show as in \cite{English.Pego.05} 
that the (negative) relative displacement profile $r(s)=\vp(s+1)-\vp(s)$ satisfies 
\begin{equation}\label{e:meq1r''}
c^2 r''(s) = F(r(s+1))-2F(r(s))+F(r(s-1)) ,  
\end{equation}
with $F(r) = \pow(1-r)^{-\pow-1}$, and infer that 
\begin{equation}
    r(s) =  (\Lambda * F\circ r)(s)= \int_{-\infty}^\infty \Lambda(s-\tau) F(r(\tau))\,d\tau, 
\end{equation}
where $\Lambda(s) = c^{-2}\max(1-|s|,0)$
That is, the function $r(\cdot)$ should be a fixed point of the nonlinear operator 
of composing with the force function $F$ followed by convolution with the `tent' function $\Lambda$.

We numerically compute profiles by straightforward spatial discretization of the following variant
of Petviashvili's iteration method for such equations~\cite{petviashvili1976equation,petviashvili2016solitary}. 
Starting with $r_0= 0.01 c^2\Lambda$, for $n=1,2,\ldots,N$ compute
\begin{align}
    \tilde r_n &= \Lambda* F\circ r_{n-1}\,, 
    \\ C_n &= \left.\int_\R r_{n-1}\middle/ \int_R \tilde r_n \right.\,, 
    \\ r_n &= C_n^q \,\tilde r_n\,. 
\end{align}
We take the exponent $q$ slightly greater than 1 to overcorrect amplitude error that otherwise 
grows with this type of iteration.
The integrals are approximated by uniform-grid discretization on the finite interval $[-20,20]$
with step size $h=0.01$.
With $N=1000$ and varying $q$ as needed, we obtain numerical convergence in all cases treated,
finding residual errors in \eqref{e:meq1r''} smaller than $10^{-12}$, and $|C_N-1|<10^{-14}$.

Nearest-neighbor profiles for a range of values of $\pow$ are plotted in Figure~\ref{fig:vary_alpha}.
In each subplot we keep the ratio of wave speed $c$ to sonic speed $c_s$ fixed, either 1.25 or 2.5. 
In Figure~\ref{fig:compare} we plot results comparing  Calogero-Moser profiles 
with profiles for the case $\beta=\pow+1=2.73$ that was used as a model for experiments in \cite{MoleronEA2014}.
(The sonic speed $c_s=\sqrt{\pow(\pow+1)}\approx 2.17322$ for \eqref{e:meq1sys} and $c_s=\pi$ for Calogero-Moser.)
For larger values of $c/c_s$ the graphs become indistinguishable, approaching the hard-collision limit in each case.
For smaller values of $c/c_s$ the Calogero-Moser profile broadens to approach a Benjamin-Ono soliton shape with
algebraic decay, while the nearest-neighbor profile approaches a scaled KdV $\sech^2$ shape,
according to results proved in \cite{Friesecke.Pego.99}.

\section*{Acknowledgements}
  This material is based upon work supported by the National Science Foundation under
  Grant No.~DMS 2106534.

\bibliographystyle{siam}
\bibliography{CM2}

\appendix

\section{Log correction to scaling in an edge case}\label{a:logkdv}

Here we prove the claim in Remark~\ref{r:log}, to the effect that in the edge case $\pow=3$
of Theorem~\ref{t:longwave} we obtain the KdV equation after a modified scaling ansatz.

Fix $\pow=3$, $p=1$, $q=3$. Let us replace the scaling ansatz in \eqref{e:scalepq} by 
\begin{equation}
    x_j = j +  \lambda\eps v(x,\tau), \qquad x = \eps(j-c_\pow t), \quad \tau =  \nu \eps^q t,
\end{equation}
where $\lambda$ and $\nu$ depend on $\eps$ in a way to be specified.

\begin{proposition}
   Under the hypotheses of Theorem~\ref{t:longwave}, if $\pow=3$ and provided
   $\lambda=\nu=\log(1/\eps)$, then
   the lattice error \eqref{e:latticeR} satisfies $R_\eps=\lambda^2\eps^5(Q+o(1))$, where 
  \[
  Q = \kappa_1 \D_\tau u + \kappa_2 u\D_x u + \D_x^3 u,
  \] 
  with $\kappa_1=2c_\pow$ and $\kappa_2 = \pow(\pow+1)(\pow+2)\zeta_\pow$ as before. 
\end{proposition}
\begin{proof}
We compute as in the proof of Theorem~\ref{t:longwave} with the following modifications.
Equations~\eqref{d:Reps} and \eqref{d:ddxj} become
\begin{align}
\label{d2:Reps}
R_\eps &= \ddot x_j+ \lambda\pow\pow_1 L_\eps +\lambda^2 \pow\pow_2 N_\eps + O(\eps^{3p+3}{\lambda^3}) \,,
\\
\label{d2:ddxj}
\ddot x_j &= -\eps^{p+2}c_\pow^2 \lambda\D_xu + 2 \eps^{p+q+1}c_\pow  \lambda\nu\D_\tau u + \eps^{p+2q}\nu^2\D_\tau^2 v \,,
\end{align}
while the expressions for $L_\eps$ and $N_\eps$ are unchanged from \eqref{d:Leps} and \eqref{d:Neps}.
The asymptotic expression for $N_\eps$ from Lemma~\ref{l:Neps} holds without change.

When we compute $L_\eps$ as in case (ii), since $(-i\sgn k)|k|^3=(ik)^3$ we find that the expression 
in \eqref{e:Lhateps} takes the form
\begin{equation} \label{e:hatL3}
    \hat L_\eps(k) = \eps^3\zeta_3\,  \widehat{\D_xu}(k) + 2\eps^5 S(\eps|k|)\, \widehat{\D_x^3u}(k)  ,
\end{equation}
where, in terms of the function $f(x) = \frac12(1-\sinc^2(x/2))/x^2$ from \eqref{d:ffrac},
\begin{equation}
S(h) = h \sum_{m=1}^\infty \frac{1-\sinc^2(mh/2)}{2(mh)^3} = \sum_{m=1}^\infty \frac1m f(mh) .
\end{equation}
\begin{lemma}
   $0<S(h)<\zeta_3/h^2$  for all $h>0$, and $S(h)\sim -\frac1{24}\log h$ as $h\to0.$
\end{lemma}
The asymptotic formula follows from L'H\^opital's rule after noting that
\[
hS'(h) = \sum_{m=1}^\infty f'(mh)h \to \int_0^\infty f'(x)\,dx = -f(0)= -\frac1{24}.
\]
From the asymptotic formula it follows $S(h)$ is {\em slowly varying} at 0, meaning 
that as $\eps\to0$, the ratio $S(\eps|k|)/S(\eps)\to1$ for all $k\ne0$.
From Karamata's theory of regular variation \cite{Seneta}, 
this limit is then uniform for $|k|$ in any compact subinterval of $(0,\infty)$,
and the ratio is $o(|k|^{-\beta})$ as $k\to0$ for any $\beta>0$, uniformly for $\eps$ small.

Using these facts in the Fourier inversion formula, 
since $\hat u(k)|k|^r\in L^1$ for all $r<\frac72$
we infer by dominated convergence that as $\eps\to0$,
\begin{equation}
    \int_\R e^{ikx}\widehat{\D_x^3 u}(k) \frac{S(\eps|k|)}{S(\eps)}\,dk \to \D_x^3 u(x,\tau).
\end{equation}
Hence 
\begin{equation}
L_\eps = \eps^3\zeta_3\D_xu - \frac{\eps^5\log \eps}{12}(\D_x^3 u + o(1)).
\end{equation}
Putting this relation into \eqref{d2:Reps} and noting $\pow\pow_1=12$, the Proposition follows.
\end{proof}

\section{KdV limit with alternating forces}\label{a:alt}
\newcommand{\mellin}[1]{\widetilde{#1}}

Here we complete the proof of Lemma~\ref{l:Spow}, which we used in Section~\ref{s:longwave} 
to establish the formal KdV limit for the system \eqref{e:sysalt} in which interaction forces alternately 
repel and attract, as in the experimental setup of Moler\'on {\em et al.}~\cite{MoleronEA2014}.

1. Fix $\pow>0$. The function $S_\pow(h)$ defined in \eqref{d:Spow} satisfies
\begin{equation}\label{e:Spowh}
S_\pow(h) = \sum_{m=1}^\infty (-1)^{m-1}f_\pow(mh) h^{\pow},
\end{equation}
where, in terms of the entire function $f$ defined in \eqref{d:ffrac},
\begin{equation}\label{d:fpow}
f_\pow(x)= \frac{1-\sinc^2(x/2)}{x^\pow} = 2x^{2-\pow}f(x).
\end{equation}
Because $f_\pow$ is eventually monotone decreasing, the alternating series~\eqref{e:Spowh}
converges uniformly on compact subsets of $(0,\infty)$, so $S_\pow$ is continuous.

2. Next we claim that the Mellin transform of $S_\pow$, given by 
\begin{equation}\label{d:mellinSpow}
\mellin{S_\pow}(s) := \int_0^\infty S_\pow(h)h^{s-1}\,ds,
\end{equation}
is well defined whenever $\max(-\pow,-1)<\re s<0$.
To control the convergence of the integral, we pair successive terms in \eqref{e:Spowh}, writing
\begin{equation}\label{e:Sh_odd}
S_\pow(h) = \sum_{m\text{ odd}} \Bigl( f_\pow(mh)-f_\pow(mh+h)\Bigr)h^\pow.
\end{equation}
We establish bounds on the terms of this sum as follows.
First, we find (by Taylor expansion for $0<x<1$) that 
since $x^4 f'(x) = -\sin x + x -4x^3f(x)$, 
\begin{align}
&f(x)\in\left(0,\frac1{4!}\right), \quad 
f'(x) \in \left(\frac{-x}{5!},0\right)
\quad \text{for $0<x\le 1$,} \quad
\\ 
& f(x)\in\left(0,\frac1{2x^{2}}\right) , \quad
f'(x)\in\left(\frac{-2}{x^3},\frac2{x^3}\right)\quad
\text{for $x>1$.} \quad 
\end{align}
Since $\frac12 f_\pow(x)=x^{2-\pow}f(x)$ we have 
$\frac12 f_\pow'(x)=(2-\pow)x^{1-\pow}f(x)+x^{2-\pow}f'(x)$, whence it follows
\begin{equation}
|f_\pow'(x)| \le 
\min(x^{-\pow},x^{-1-\pow})\cdot (6+\pow) =:\lambda_\pow(x).
\end{equation}
The function $\lambda_\pow(x)$ is (chosen to be) decreasing on $(0,\infty)$, ensuring that
\[
|f_\pow(x)-f_\pow(x+h)|\le \lambda_\pow(x)h  \quad\text{for all $x,h>0$}.
\]
Applying this estimate in \eqref{e:Sh_odd}, we note that by 
Tonelli's theorem,
\begin{align*}
  \int_0^\infty|S_\pow(h)| h^{s-1}\,dh
    &\le  
    \int_0^\infty\sum_{m=1}^\infty \lambda_\pow(mh) h^{\pow+s}\,dh
    \\ &= \sum_{m=1}^\infty \frac1{m^{\pow+s+1}}\int_0^\infty \lambda_\pow(x)x^{\pow+s}\,dx
    \\ &= \zeta(\pow+s+1) \int_0^\infty C \min(1,x^{-1})x^s\,dx,
\end{align*}
and this is {\em finite}  provided $0<\pow+s$ and $-1<s<0$. 
Thus we find that the Mellin transform of $S_\pow$ is 
well defined as claimed, for $\re s\in(\max(-\pow,-1),0)$.

3. After use of Fubini's theorem and change of variables, we compute that
\begin{align}
\mellin{S_\pow}(s) 
&= \sum_{m\text{ odd}} \int_0^\infty
\Bigl(f_\pow(mh)-f_\pow(mh+h)\Bigr)h^{\pow+s-1}\,dh 
\nonumber \\&=  
\sum_{m=1}^\infty \frac{(-1)^{m-1}}{m^{\pow+s}} \int_0^\infty f_\pow(x) x^{\pow+s-1}\,dx
\nonumber \\&= \zeta^*_{\pow+s}\, \mellin{f_\pow}(\pow+s).
\end{align}
Recall that in Remark~\ref{r:eta} we effectively computed the Mellin transform of $f$. 
From \eqref{d:fpow} and \eqref{e:Is0}, we find that for $-2<\re s<0$,
\[
\mellin{f_\pow}(\pow+s) = 2\mellin{f}(2+s) = 2 \cos\bigl(\tfrac{\pi}2(s-2)\bigr) \Gamma(s-2).
\]
Thus $\mellin{S_\pow}$ extends to a meromorphic function on $\C$ with simple poles at $s=2-2k$ for 
$k=0,1,2,\ldots$, where the residue is 
\begin{equation}
    {\rm Res}(\mellin{S_\pow},2-2k) = 2\zeta^*_{\pow+2-2k}\frac{(-1)^k}{(2k)!} .
\end{equation}
4. Choose $\delta\in(0,\min(1,\pow))$. 
We claim that, for $\sigma\in(-4,-\delta)$ we have
\begin{align}\label{b:Sp}
    |\mellin{S_\pow}(\sigma+it)| \to 0 &\quad\text{as $|t|\to\infty$, uniformly in $\sigma$, and}
    \\ t\mapsto |\mellin{S_\pow}(\sigma+it)| &\quad\text{is integrable on $(-\infty,\infty)$ if $\sigma\ne-2.$}
    \label{b:Spint}
\end{align}
Using the asymptotic formula~\cite[\href{https://dlmf.nist.gov/5.11.E9}{(5.11.9)}]{NIST:DLMF} 
\[
|\Gamma(\sigma+it)| \sim 
    \sqrt{2\pi}|t|^{\sigma-(1/2)}e^{-\pi|t|/2} ,
\]
which is valid uniformly as $t\to\pm\infty$ for $\sigma$ real and bounded,
for $|t|>1$ we get
\begin{equation}\label{e:fpowest}
|\mellin{f_\pow}(\pow+\sigma+it)| \le C |t|^{\sigma-(5/2)}.
\end{equation}
Further, since $|\zeta^*_s|=O(|\zeta_s|)$, 
the zeta-function bounds from \cite[(5.1.1)]{Titchmarsh-zeta} yield  
\begin{align}
    |\zeta^*_{\pow+\sigma+it}| = \begin{cases}
    O(|t|^{(1/2)-\pow-\sigma}) & \text{for }\sigma\le -\delta-\pow ,\\
    O(|t|^{(3/2)+\delta}) & \text{for }\sigma\ge -\delta-\pow .
    \end{cases}
\end{align}
We infer that \eqref{b:Sp}--\eqref{b:Spint} hold, and in particular,
\begin{equation}
|\mellin{S_\pow}(\sigma+it)| = 
\begin{cases}
O(|t|^{-2-\pow}) & \text{for }\sigma\le-\delta-\pow, \\
O(|t|^{-1+\delta+\sigma}) & \text{for }\sigma >-\delta-\pow.
\end{cases}
\end{equation}

5.  By the Mellin inversion theorem (see \cite[Thm. 2]{Flajolet} and \cite[Thm.~8.26]{Folland})
thus we have
\begin{equation}\label{e:iMellin}
S_\pow(h) = \frac1{2\pi i}\lim_{T\to\infty}\int_{c-iT}^{c+iT} 
\mellin{S_\pow}(s)h^{-s}\,ds , 
\quad\text{for $c\in(\max(-\pow,-1),-\delta).$}
\end{equation}
Now, because of the uniform decay in \eqref{b:Sp},
we can deform the path in \eqref{e:iMellin} to move $c$ from the interval stated
to a value $c\in(-4,-2)$, picking up only the residue of the integrand at
the pole $s=-2$ ($k=2$).  Thus we find
\[
S_\pow(h) = \tfrac1{12}\zeta^*_{\pow-2}\,h^2 +  \tilde E(h),
\quad
\tilde E(h) = \frac{h^{-c}}{2\pi} \int_{-\infty}^\infty \mellin{S_\pow}(c+it) h^{-it}\,dt.
\]
Because $t\mapsto |\mellin{S_\pow}(c+it)|$ is integrable for such $c$, we find
$\tilde E(h) = O(h^{-c})$ as $h\to0$. 
With $c=-3$ this yields the desired statement,
and finishes the proof.

\section{B\"acklund transform for the Calogero-Suther\-land system}\label{a:backCS}

For the convenience of the reader,  we verify here that the (trigonometric) B\"acklund 
transform equations \eqref{e:perback1}--\eqref{e:perback2} described by Wojciechowski \cite{Wojciechowski.82} 
imply the Calogero-Moser-Sutherland equations \eqref{e:CSx}--\eqref{e:CSy}
in the case that we use, corresponding to the trigonometric pair potential.

It is efficient to follow the approach suggested for the rational case in \cite{CM-mathoverflow}
and write the equations as one system involving 
variables $x_1,\ldots,x_{N+M}$ and signs $\sigma_1,\ldots,\sigma_{N+M}$ defined by 
\begin{align}
    x_j &= y_{j-N}, \quad j=N+1,\ldots,N+M, 
    \\ \sigma_j &= \begin{cases}
        +1 & j=1,\ldots,N,\\
        -1 & j=N+1,\ldots,M.
    \end{cases}
\end{align}
With these variables, the B\"acklund pair \eqref{e:perback1}--\eqref{e:perback2} 
is written in a unified way as 
\begin{align} \label{e:b_unified}
    i\dot x_j = 
    a \sum_{\substack{k=1\\k\ne j}}^{N+M} \sigma_k \cot a(x_j-x_k) = 
    ia \sum_{\substack{k=1\\k\ne j}}^{N+M} \sigma_k \frac{z_j+z_k}{z_j-z_k} \,,
\end{align}
with  $z_j = e^{2iax_j}$, $j=1,\ldots,N+M$.
Noting $\cot x \sin^{-2}x = \cot x(1+\cot^2 x)$
and using the cotangent difference identity in \eqref{e:cot_ident},
in order to prove that equations \eqref{e:CSx}--\eqref{e:CSy} hold,
our goal is to show that 
\begin{align} \label{e:back_goal}
    \ddot x_j &= 2a^3 \sum_{\substack{k\ne j\\ \sigma_k=\sigma_j}}  
    i\frac{z_j+z_k}{z_j-z_k}
    \left( 1- \left(\frac{z_j+z_k}{z_j-z_k}\right)^2\right)
  \nonumber  \\& = -8ia^3 \sum_{\substack{k\ne j\\ \sigma_k=\sigma_j}}  
    \frac{z_j+z_k}{(z_j-z_k)^3} z_jz_k , \quad j=1,\ldots,N+M.
\end{align}
Toward this end, we note \eqref{e:b_unified} implies
\begin{equation}\label{e:dotzj}
    \dot z_j = 2ia \dot x_j z_j = 2ia^2 z_j \sum_{l\ne j} \sigma_l \frac{z_j+z_l}{z_j-z_l}  \,.
\end{equation}
We then differentiate \eqref{e:b_unified} and use \eqref{e:dotzj} to find
\begin{align}
    \ddot x_j &= a\sum_{k\ne j} \sigma_k 
    \frac{(\dot z_j+\dot z_k)(z_j-z_k)-(z_j+z_k)(\dot z_j-\dot z_k)}
    {(z_j-z_k)^2}
    \\ 
    &= 2a \sum_{k\ne j} \sigma_k
    \frac{z_j\dot z_k - z_k \dot z_j}
    {(z_j-z_k)^2}
   \\  &= 4ia^3 \sum_{k\ne j} \frac{\sigma_k z_jz_k}{(z_j-z_k)^2}
   \left(
   \sum_{l\ne k} \sigma_l \frac{z_k+z_l}{z_k-z_l}
   -\sum_{l\ne j} \sigma_l \frac{z_j+z_l}{z_j-z_l}
   \right ).
\end{align}
Extracting the terms with $l=j$ and $l=k$ from the respective inner sums yields
\begin{align}
    \ddot x_j &= 4ia^3 \sum_{k\ne j} \frac{\sigma_k z_jz_k}{(z_j-z_k)^2}
    \left( \sigma_j \frac{z_k+z_j}{z_k-z_j} - \sigma_k \frac{z_j+z_k}{z_j-z_k} \right)
    \nonumber \\
    & \quad + \, 4ia^3 \sum_{k\ne j} \sum_{l\notin\{j,k\}} \frac{\sigma_k z_jz_k}{(z_j-z_k)^2}
    \cdot\frac{2\sigma_l z_l (z_j-z_k)}{(z_k-z_l)(z_j-z_l)} .
\end{align}
The double sum vanishes because terms cancel in pairs: upon switching $k$ and $l$ the factors
$\sigma_kz_k\sigma_lz_l$ and $(z_j-z_k)(z_j-z_l)$ remain the same, but $(z_k-z_l)$ changes sign.
Therefore upon simplifying, since $\sigma_k^2=1$ we find
\begin{equation}
 \ddot x_j = -4ia^3 \sum_{k\ne j} (\sigma_k\sigma_j+1)\frac{z_j+z_k}{(z_j-z_k)^3}z_jz_k\,.
\end{equation}
This is equivalent to \eqref{e:back_goal}, finishing the proof that 
the (trigonometric) B\"acklund equations~\eqref{e:perback1}--\eqref{e:perback2}
imply the Calogero-Sutherland equations~\eqref{e:CSx}--\eqref{e:CSy}.

\section{Projection method for Calogero-Sutherland systems}\label{a:projection_method}

For completeness, here we provide a brief account of the projection method
of Olshanetsky and Perelomov~\cite{Olshanetsky.Perelomov.81} 
for the solution of Calogero-Sutherland system. 
There is a singularity in this case of the trigonometric potential
that these authors did not explicitly treat, and we also describe a modified 
solution procedure that does not encounter any singularity.

The calculations are rather straightforward if one takes for granted, as we do, 
the result of Moser~\cite{Moser.75} which states that for any solution of the Calogero-Sutherland
equations in \eqref{e:CSx}, the Lax equation
\begin{equation}
    \dot L + [iM,L] = 0
\end{equation}
holds, with time-dependent matrices $L$ and $M$ given by 
\begin{align*}
    L_{jk} &= 
    \delta_{jk}\, \dot x_j + 
    (1-\delta_{jk}) ia \cot a(x_j-x_k) \,,
\\ iM_{jk} &=
    \delta_{jk} \sum_{l\ne j} ia^2\csc^2 a(x_j-x_k)
    - (1-\delta_{jk}) ia^2 \csc^2 a(x_j-x_k)  \,.
\end{align*}
With the (unitary) matrix $U(t)$ determined by solving $\dot U = U(iM)$, $U(0)=I$, 
the matrix $V=ULU\inv$ is constant in time, since $\dot V = U(\dot L + [iM,L])U\inv=0$.
The idea of Olshanetsky and Perelomov is to define a unitary matrix
\begin{equation}\label{d:Y}
    Y(t) = U Z U\inv, \qquad Z = \diag(z_j) = \diag(e^{2iax_j})
\end{equation}
and compute $\dot Y= U(\dot Z + [iM,Z])U\inv$, and further that 
\begin{align*}
    \tfrac12(\dot YY\inv + Y\inv \dot Y) &= 
    U\left(\dot Z Z\inv + \tfrac12(Z\inv(iM)Z - Z(iM)Z\inv)\right)U\inv
    \\ &= U(2ia L)U\inv = 2ia V .
\end{align*}
With $B_{jk}=\delta_{jk}e^{ia x_j(0)}$, the matrix $X(t)=B e^{2iaAt}B$ satisfies
$X(0)=Y(0)$ and  
\[
\tfrac12 (\dot X X\inv + X\inv \dot X) = 2iaV \,,
\]
the same differential equation, {\em provided}
\[
 \tfrac12 (BAB\inv + B\inv AB) =V.
\]
Since $V=L(0)$, this means that $A_{jj}=L_{jj}=\dot x_j(0)$, 
while for $j\ne k$, with  $\zeta_j=e^{iax_j}$ and at $t=0$, 
\[
\tfrac12 A_{jk} (\zeta_j\bar\zeta_k+\bar\zeta_j\zeta_k) =  A_{jk} \frac{z_j+z_k}{2\zeta_j\zeta_k} = L_{jk} = -a \frac{z_j+z_k}{z_j-z_k}\,.
\]
{\em Provided} $z_j+z_k\ne0$ for all $j\ne k$, it follows
\begin{equation}\label{ae:Ajk}
A_{jk} = \frac{-2\zeta_j\zeta_k}{z_j-z_k} = \frac{ia}{\sin a(x_j-x_k)},
\end{equation}
and it follows $X(t)=Y(t)$ for all $t$.
A subtlety here is that $z_j+z_k=0$ is possible and corresponds to a singular set
for the differential equation on the unitary group where $\dot X$ is not determined.  
However, one can approximate by nonsingular data and justify that \eqref{ae:Ajk} yields $Y(t)=X(t)$ in all cases.
The result is that $z_1(t),\ldots,z_N(t)$ are the eigenvalues of $X(t)$ for all $t$.

{\em Modified procedure.}
A modified solution procedure that has no singularities is as follows.
Using \eqref{d:Y}, compute instead that 
\begin{align}
\nonumber
\dot Y Y\inv &= U(\dot Z Z\inv + iM - Z(iM)Z\inv) U\inv
\\ &=  U(2ia(L+K))U\inv =:2ia W,
\end{align}
where $2ia K = iM - \frac12 Z\inv(iM)Z -\frac12 Z(iM)Z\inv$. Evidently $K_{jj}=0$ for all $j$,
while for $j\ne k$,
\begin{align*}
2ia K_{jk} &= i M_{jk}(1-\tfrac12(\bar z_j z_k+z_j\bar z_k))
\\ &= \frac{-ia^2}{\sin^2 a(x_j-x_k)} (1-\cos 2a(x_j-x_k)) = -2ia^2.
\end{align*}
Hence $K = a(I-uu^T)$, where $u_j=1$ for all $j$. Since $\dot K=0$, and because $Mu=0$ and $M$ is Hermitian, 
we have $[iM,K]=0$, and we find
\[
\dot W = U(\dot L + [iM,L] + \dot K + [iM,K])U\inv = 0.
\]
Therefore, since $Y(0) = Z(0)$ and $\dot Y = (2iaW)Y$, it follows 
\begin{equation}
Y(t) = e^{2iaWt}Z(0), 
\end{equation}
where the entries of the constant matrix $W$ are 
\begin{equation}
W_{jk} = L_{jk} + K_{jk} = 
    \delta_{jk}\, \dot x_j + 
    (1-\delta_{jk}) (ia \cot a(x_j-x_k) - a) \,,
\end{equation}
evaluated at $t=0$. The time-dependent values $z_1(t),\ldots,z_N(t)$ are again the eigenvalues of $Y(t)$.
\end{document}